\documentclass[letterpaper, 10 pt, conference]{IEEEtran}
\makeatletter
\def\ps@headings{%
\def\@oddhead{\mbox{}\scriptsize\rightmark \hfil \thepage}%
\def\@evenhead{\scriptsize\thepage \hfil\leftmark\mbox{}}%
\def\@oddfoot{}%
\def\@evenfoot{}}
\makeatother
\usepackage[T1]{fontenc}
\usepackage[latin1]{inputenc}
\usepackage{float}
\usepackage{amsmath}
\usepackage[noadjust]{cite}
\usepackage{amssymb}
\usepackage{enumerate}
\usepackage{amsthm}
\usepackage{graphics,epsfig}
\usepackage{subfigure,epsf,pstricks}

\floatstyle{ruled}
\newfloat{algorithm}{tbp}{loa}
\floatname{algorithm}{Algorithm}

\usepackage{algorithmic}
\algsetup{linenodelimiter= }

\title{The Limits of Error Correction with $l_p$ Decoding}

\author{Meng Wang \ \ \ \ \ Weiyu Xu \ \ \ \ \ Ao Tang
\\
School of ECE, Cornell University, Ithaca, NY 14853, USA }

\usepackage{graphicx}
\usepackage{epstopdf}
\newtheorem{theorem}{Theorem}
\newtheorem{lemma}{Lemma}
\newtheorem{cor}{Corollary}

\newtheorem{prop}{Proposition}

\begin{document}

\maketitle \thispagestyle{empty} \pagestyle{empty}

\begin{abstract}
An unknown vector $f$ in $\mathbf{R}^n$ can be recovered from
corrupted measurements $y=Af+e$ where $A^{m \times n}$($m \geq n$)
is the coding matrix if the unknown error vector $e$ is sparse. We
investigate the relationship of the fraction of errors and the
recovering ability of $l_p$-minimization ($0<p \leq 1$) which
returns a vector $x$ minimizing the ``$l_p$-norm'' of $y-Ax$. We
give sharp thresholds of the fraction of errors that determine the
successful recovery of $f$. If $e$ is an arbitrary unknown vector,
the threshold strictly decreases from 0.5 to 0.239 as $p$ increases
from 0 to 1. If $e$ has fixed support and fixed signs on the
support, the threshold is $\frac{2}{3}$ for all $p$ in $(0,1)$,
while the threshold is 1 for $l_1$-minimization.
\end{abstract}

\section{Introduction} \label{sec:intro}
We consider recovering a vector $f$ in $\mathbf{R}^n$ from corrupted
measurements $y=Af+e$, where $A^{m \times n}$($m \geq n$) is the
coding matrix and $e$ is an arbitrary and unknown vector of errors.
Obviously, if the fraction of the corrupted entries is too large,
there is no hope of recovering $f$ from $Af+e$. However, if the
fraction of corrupted measurements is small enough, one can actually
recover $f$ from $y=Af+e$. As the sparsity of $e$ is represented by
the $l_0$ norm, $\|e\|_0:=|\{i:e_i \neq 0\}|$, one natural way is to
find a vector $x$ such that the number of terms where $y$ and $Ax$
differ is minimized. Mathematically, we solve the following
$l_0$-minimization problem:
\begin{equation}\label{eqn:l0}
\min \limits_{x \in \mathbf{R}^n} \|y-Ax \|_0.
\end{equation}
However, (\ref{eqn:l0}) is combinatorial and computationally
intractable, and one commonly used approach is to solve a closely
related $l_1$-minimization problem:

\begin{equation}\label{eqn:l1}
\min \limits_{x \in \mathbf{R}^n} \|y-Ax \|_1 
\end{equation}
where $\|x\|_1:= \sum_i |x_i|$. (\ref{eqn:l1}) can be recast as a
linear program, thus can be solved efficiently. Conditions under
which (\ref{eqn:l1}) can successfully recover $f$ have been
extensively studied in the literature of compressed sensing
(\cite{DoT05, Don06, CaT05,CaT06,SXH08,WM08}). For example,
\cite{CaT05} gives a sufficient condition known as the Restricted
Isometry Property (RIP).

Recently, there has been great research interest in recovering $f$
by $l_p$-minimization for $p<1$
(\cite{Chartrand07,Chartrand072,SCY08,DG09,FL09}) as follows,
\begin{equation}\label{eqn:lp}
\min \limits_{x \in \mathbf{R}^n} \|y-Ax \|_p^p.
\end{equation}
Recall that $ \|x\|_p^p:=\left(\sum_i |x_i|^p\right)$ for $p>0$. We say $f$ can be recovered by $l_p$-minimization if and only if it
is the unique solution to (\ref{eqn:lp}). 
%
%
%
Then the question is what is the relationship between the sparsity
of the error vector and the successful recovery with
$l_p$-minimization?
(\ref{eqn:lp}) is non-convex, and thus it is generally hard to
compute the global minimum. However, \cite{Chartrand07} shows
numerically that we can recover $f$ by finding a local minimum of
(\ref{eqn:lp}), and $l_p$-minimization outperforms
$l_1$-minimization in terms of the sparsity restriction for $e$.
\cite{SCY08} extends RIP to $l_p$-minimization and analyzes the
ability of $l_p$-minimization to recover signals from noisy
measurements. \cite{FL09} also provides a condition for the success
recovery via $l_p$-minimization, which can be generalized to $L_1$
case. Both conditions are sufficient but not necessary, and
thus are too restrictive in general. 


Let $e \in \mathbf{R}^m$ be an arbitrary and unknown vector of
errors on support $T=\{i: e_i \neq 0 \}$. We say $e$ is $\rho
m$-sparse if $|T| \leq \rho m$ for some $\rho<1$ where $|T|$ is the
cardinality of set $T$. Our main contribution is a sharp threshold
$\rho^*(p)$ for all $p \leq 1$ such that for $\rho< \rho^*(p)$, if
$m\geq Cn$ for some constant $C$ and the entries of $A$ are i.i.d.
Gaussian, then $l_p$-minimization can recover $f$
with overwhelming probability.  
We provide two thresholds: one ($\rho^*$) is for the case when $e$ is an arbitrary unknown vector, and the other ($\rho_w^*$) assumes that $e$ has fixed support and fixed signs. 
In the latter case, the condition of successful recovery with
$l_1$-minimization from any possible error vector is the same, while
the condition of successful recovery with $l_p$-minimization ($p<1$)
from different error vectors differs. Using worst-case performance
as criterion, we prove that though $l_p$ outperforms $l_1$ in the
former case, it is not comparable to $l_1$ in the latter case. Both
bounds $\rho^*$ and $\rho_w^*$ are tight in the sense that once the
fraction of errors exceeds $\rho^*$ (or $\rho^*_w$),
$l_p$-minimization can be made to fail with overwhelming
probability. Our technique stems from \cite{DMT07}, which only
focuses on $l_1$-minimization and the case that $e$ is arbitrary. 
\section{Recovery From Arbitrary Error vector}\label{sec:sbd}
In this section, we shall give a function $\rho^*(p)$ such that for
a given $p$, for any $\rho<\rho^*(p)$, when the entries of $A$ are
i.i.d. Gaussian, the $l_p$-minimization can recover $f$ with
overwhelming probability as long as the error $e$ is $\rho
m$-sparse.

The following theorem gives an equivalent condition for the success
of $l_p$ minimization (~\cite{Chartrand07,Chartrand072}). 

\begin{theorem}[~\cite{Chartrand07,Chartrand072}]\label{thm:slp}
$f$ is the unique solution to $l_p$ minimization problem $(0 <p \leq
1)$ for every $f$ and for every $\rho m$-sparse $e$ if and only if
\begin{equation}\label{eqn:slp}
\sum \limits_{i \in T} |(Az)_i|^p < \sum \limits_{i \in T^c} |(Az)_i|^p
\end{equation}
for every $z \in \mathbf{R}^n$, and every support $T$ with $|T| \leq
\rho m$.
\end{theorem}
%

One important property is that if the condition (\ref{eqn:slp}) is
satisfied for some $0<p \leq 1$, then it is also satisfied for all
$0 < q \leq p$ (\cite{DG09}). Now we define the threshold of
successful recovery $\rho^*$ as a function of $p$.
\begin{lemma}
Let $X_1$, $X_2$,...,$X_m$ be i.i.d $N(0,1)$ random variables and
let $Y_1$, $Y_2$,...,$Y_m$ be the sorted ordering (in non-increasing
order) of $|X_1|^p$, $|X_2|^p$,...,$|X_m|^p$ for some $p \in (0,1]$.
For a $\rho
>0$, define $S_\rho$ as $\sum \limits_{i=1}^{\lceil \rho m \rceil}
Y_i$. Let $S$ denote $E[S_1]$, the expected value of $S_1$. Then
there exists a constant $\rho^*(p)$ such that $\lim \limits_{m
\rightarrow \infty} \frac{E[S_{\rho^*}]}{S}=\frac{1}{2}$.
\end{lemma}
\begin{proof}
Let $X \sim N(0,1)$ and let $Z=|X|$. Let $f(z)$ denote the p.d.f. of
$Z$ and $F(z)$ be its c.d.f. Define $g(t)= \int_{t}^ \infty z^p
f(z)dz$. $g$ is continuous and decreasing in $[0, \infty]$, and
$g(0)=E[Z^p]=\frac{S}{m}$, $\lim _{t \rightarrow \infty} g(t)=0$.
Then there exists $z^*$ such that $g(z^*)=\frac{g(0)}{2}$, we claim
that $\rho^*=1-F(z^*)$ has the desired property.

Let $T_t=\sum_{i:Y_i \geq t^p}Y_i$. Then $E[T_{z^*}]=mg(z^*)$. Since
$E[|T_{z^*}-S_{\rho^*}|]$ is bounded by $O(\sqrt{m})$, and
$S=mg(0)$, thus $\lim _{m \rightarrow \infty}
\frac{E[S_{\rho^*}]}{S}=\frac{1}{2}.$

%
\end{proof}

\begin{prop}\label{prop:rho}
The function $\rho^*(p)$ is strictly decreasing 
in $p$ on $(0,1]$.
\end{prop}
\begin{proof}
From the definition of $z^*$ and $\rho^*(p)$, we have
\begin{equation} \label{eqn:zp}
H(z^*,p):=\int_0^{z^*}x^pf(x)dx -\int_{z^*}^\infty x^pf(x)dx=0,
\end{equation}
and
\begin{equation} \nonumber
 \rho^*=1-F(z^*),
 \end{equation}
  where $f(\cdot)$
and $F(\cdot)$ are the p.d.f. and c.d.f. of $|X|$, $X \thicksim
N(0,1)$.

From the Implicit Function Theorem,
\begin{equation}\nonumber
\frac{dz^*}{dp}=-\frac{\frac{\partial H}{\partial p}}{\frac{\partial
H}{\partial z^{*}}}=-\frac{\int_0^{z^*} x^p(\ln x)f(x)dx
-\int_{z^*}^\infty x^p(\ln x) f(x)dx}{2z^{*p}f(z^*)}
\end{equation}

From the chain rule, we know
$\frac{d\rho^*}{dp}=\frac{d\rho^*}{dz^*}\frac{dz^*}{dp}$, thus
\begin{equation}\label{eqn:drhodp}
\frac{d\rho^*}{dp}
=\frac{\int_0^{z^*} x^p(\ln x)f(x)dx -\int_{z^*}^\infty x^p(\ln x) f(x)dx}{2z^{*p}}
\end{equation}


Note the numerator of (\ref{eqn:drhodp}) is less than 0 from
(\ref{eqn:zp}), thus $\frac{d\rho^*}{dp}<0$.

\end{proof}

We plot $\rho^*$
against $p$ numerically in Fig. \ref{fig:rho}. $\rho^*(p)$ goes to
$\frac{1}{2}$ as $p$ tends to zero. Note that $\rho^*(1)=0.239...$,
which coincides with the result in \cite{DMT07}.

\begin{figure}[t]
      \centering
      \includegraphics[scale=0.5]{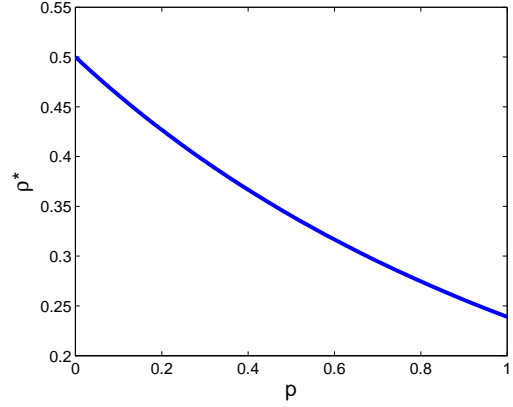}
 \caption{Threshold $\rho^*$ of successful recovery with $l_p$-minimization}
   \label{fig:rho}
\end{figure}

Now we proceed to prove that $\rho^*$ is the threshold of successful
recovery with $l_p$ minimization for $p$ in $(0,1]$. First we state
the concentration property of $S_\rho$ in the following lemma.
\begin{lemma}\label{lemma:srho}
For any $p \in (0,1]$, let $X_1$,...,$X_m$, $Y_1$,...,$Y_m$,
$S_\rho$ and $S$ be as above. For any $\rho>0$ and any $\delta>0$,
there exists a constant $c_1>0$ such that when $m$ is large enough,
with probability at least $1-2e^{-c_1m}$, $|S_\rho-E[S_\rho]| \leq
\delta S$.
\end{lemma}

\begin{proof}
Let $X=[X_1,...,X_m]^T$. If two vectors $X$ and $X'$ only differ in
co-ordinate $i$, then for any $p$, $|S_\rho(X)-S_\rho(X')| \leq
||X_i|^p-|X'_i|^p|$. Thus for any $X$ and $X'$,
\begin{equation}\nonumber
|S_\rho(X)-S_\rho(X')| \leq \sum_{i: X_i \neq X'_i}
\big||X_i|^p-|X'_i|^p\big| = \sum_{i} \big||X_i|^p-|X'_i|^p\big|.
\end{equation}
Since $\big||X_i|^p-|X'_i|^p\big| \leq |X_i-X'_i|^p$ for all $p \in
(0,1]$,
\begin{equation}\label{eqn:srho}
|S_\rho(X)-S_\rho(X')|  \leq \sum_{i} |X_i-X'_i|^p.
\end{equation}

 From
 the isoperimetric inequality for the Gaussian measure (\cite{Ledoux01}), for any set
 $A$ with measure at least a half, the set $A_t=\{x \in
 \mathbf{R}^m: d(x,A) \leq t\}$ has measure at least
 $1-e^{-t^2/2}$, where $d(x,A)= \inf_{y \in A} \|x-y\|_2$.
  Let $M_\rho$ be the median value of $S_\rho=S_\rho (X)$. Define set $A=\{x \in \mathbf{R}^m: S_\rho(x) \leq M_\rho\}$, then
\begin{equation}\nonumber
Pr[d(x,A) \leq t]\geq 1-e^{-t^2/2}.
\end{equation}
We claim that $d(x,A) \leq t$ implies that $S_\rho(x) \leq
M_\rho+m^{(1-p/2)}t^p$. If $x \in A$, then $S_\rho(x) \leq M_\rho$,
thus the claim holds as $m^{1-p/2}t^p$ is non-negative. If $x \notin
A$, then there exists $x' \in A$ such that $\|x-x'\|_2 \leq t$. Let
$u_i=1$ for all $i$ and let $v_i =|x_i-x'_i|^p$. From H\"older's
inequality
\begin{eqnarray}
\sum_{i} |x_i-x'_i|^p  
& \leq & \left( \sum_{i}|u_i|^{2/(2-p)}\right)^{1-p/2}  \left( \sum_{i} |v_i|^{2/p}\right)^{p/2} \nonumber \\
&\leq & m^{(1-p/2)} (t^2)^{p/2} = m^{(1-p/2)}t^p \label{eqn:holder}
\end{eqnarray}

From (\ref{eqn:srho}) and (\ref{eqn:holder}),
$|S_\rho(x)-S_\rho(x')| \leq m^{(1-p/2)}t^p$. Since $ x \notin A$
and $x' \in A$, then $S_\rho(x) > M_\rho \geq S_\rho(x')$. Thus
$S_\rho(x) \leq M_\rho+m^{(1-p/2)}t^p$, which verifies our claim.
Then
\begin{equation}\label{eqn:srholb}
Pr[S_\rho(x) \leq M_\rho+m^{(1-p/2)}t^p]\geq Pr[d(x,A) \leq t]\geq 1-e^{-t^2/2}.
\end{equation}
Similarly,
\begin{equation}\label{eqn:srhoub}
Pr[S_\rho(x) \geq M_\rho-m^{(1-p/2)}t^p]\geq 1-e^{-t^2/2}.
\end{equation}
Combining (\ref{eqn:srholb}) and (\ref{eqn:srhoub}),
\begin{equation}\label{eqn:srhomrho}
Pr[|S_\rho(x)-M_\rho| \geq m^{(1-p/2)}t^p]\leq 2e^{-t^2/2}.
\end{equation}

The difference of $E[S_{\rho}]$ and $M_{\rho}$ can be bounded as
follows,
\begin{eqnarray*}
|E[S_\rho]-M_\rho|
&\leq& E[|S_\rho-M_\rho|]\\
&=&\int_0^\infty Pr[|S_\rho(x)-M_\rho| \geq y] dy\\
&\leq &\int_0^\infty 2e^{-\frac{1}{2}y^{\frac{2}{p}}m^{(1-\frac{2}{p})}} dy\\
&=&m^{(1-\frac{p}{2})}\int_0^\infty 2e^{-\frac{1}{2}s^{\frac{2}{p}}}ds\\
\end{eqnarray*}

Note that $c:=\int_0^\infty 2e^{-\frac{1}{2}s^{(2/p)}}ds$ is a
finite constant for all $p \in (0,1]$. As $p>0$ and $S=mE[|x_i|^p]$,
thus for any $\delta >0$, $cm^{(1-\frac{p}{2})} <\frac{\delta}{2}S$
when $m$ is large enough.

Let $t=\left(\frac{1}{2}\delta S
m^{(\frac{p}{2}-1)}\right)^\frac{1}{p}=(\frac{1}{2}\delta
E[|x_i|^p])^{\frac{1}{p}}\sqrt{m}$, from (\ref{eqn:srhomrho}) with
probability at least ($1- 2e^{-\frac{1}{2}(\frac{1}{2}\delta
E[|x_i|^p])^{\frac{2}{p}}m}$), $|S_\rho-M_\rho| < \frac{1}{2}\delta
S$. Thus $|S_\rho-E[S_\rho]| \leq
|S_\rho-M_\rho|+|M_\rho-E[S_\rho]|< \delta S$ with probability at
least $1-2e^{-c_1m}$ for some constant $c_1$.
\end{proof}
\begin{cor}\label{cor:srho}
For any $\rho < \rho^*$, there exists a $\delta>0$ and a constant
$c_2>0$ such that when $m$ is large enough, with probability
$1-2e^{-c_2m}$, $S_\rho \leq (\frac{1}{2}-\delta)S$.
\end{cor}
\begin{proof}
When $\rho < \rho^*$,
\begin{eqnarray*}
E[S_\rho]&=&E[S_{\rho^*}]-\sum \limits_{i=\lceil \rho m\rceil+1}^{\lceil \rho^*m \rceil} E[|X_i|^p]\\
& \leq & E[S_{\rho^*}]-(\lceil \rho^*m \rceil-\lceil \rho m\rceil) E[|X_i|^p]
\end{eqnarray*}
Then $E[S_\rho]/S \leq \frac{1}{2}-2\delta$ for a suitable $\delta$
as $S=mE[|X_i|^p]$. The result follows by combining the above with
Lemma \ref{lemma:srho}.
\end{proof}

\begin{cor}\label{cor:s1}
For any $\epsilon >0$, there exists a constant $c_3>0$ such that
when $m$ is large enough, with probability $1-2e^{-c_3m}$, it holds
that $(1-\epsilon)S \leq S_1\leq (1+\epsilon)S$.
\end{cor}

The above two corollaries indicate that with overwhelming
probability the sum of the largest $\lceil \rho m \rceil $ terms of
$Y_i$'s is less than half of the total sum $S_1$ if $\rho< \rho^*$.
The following lemma extends the result to every vector $Az$ where
matrix $A^{m \times n}$ has i.i.d. Gaussian entries and $z$ is any
vector in $\mathbf{R}^n$.
\begin{lemma}\label{lemma:slp}
For any $0<p \leq 1$, given any $\rho < \rho^*(p)$, there exist
constants $c_4$, $c_5$, $\delta >0$ such that when $m \geq c_4 n$
and $n$ is large enough, with probability $1-e^{-c_5n}$, an $m
\times n$ matrix $A$ with i.i.d. $N(0,1)$ entries has the following
property: for every $z \in \mathbf{R}^n$ and every subset $T
\subseteq \{1,...,m\}$ with $|T| \leq \rho m$, $\sum \limits_ {i \in
T^c} |(Az)_i|^p - \sum \limits_{i \in T} |(Az)_i|^p \geq \delta S
\|z\|_2^p$.
\end{lemma}

\begin{proof}
For any given $\gamma>0$, there exists a $ \gamma$-net $K$ of
cardinality less than $(1+\frac{2}{\gamma})^n$(\cite{Ledoux01}).  A
$ \gamma$-net $K$ is a set of points such that $\|v^k\|_2=1$ for all
$v^k$ in $K$ and for any $z$ with $\|z\|_2=1$, there exists some
$v^k$ such that $\|z-v^k\|_2 \leq \gamma$. 

Since $A$ has i.i.d $N(0,1)$ entries, then $Av^k$ has $m$ i.i.d.
$N(0,1)$ entries. Applying a union bound to Corollary \ref{cor:srho}
and \ref{cor:s1}, we know that for some $\delta>0$ and for every
$\epsilon
>0$, with probability $1-2e^{-cm}$ for some $c>0$, we have
\begin{equation}\label{prop:srho}
S_\rho(Av^k) \leq
(\frac{1}{2}-\delta)S
\end{equation} and
\begin{equation}\label{prop:s1}
(1-\epsilon)S \leq S_1(Av^k)\leq (1+\epsilon)S
\end{equation}
hold for a vector $v^k$ in $K$.
 Taking $m=c_4n$ for large enough $c_4$, from union bound we get
 that (\ref{prop:srho}) and (\ref{prop:s1}) hold for all the points in $K$ at the same time with probability at least $1-e^{-c_5n}$ for some $c_5>0$.

 For any $z$ such that $\|z\|_2=1$, there exists $v_0$ in $K$ such
 that $\|z-v_0\|_2\triangleq \gamma_1 \leq \gamma$. Let $z_1$ denote $z-v_0$,
 then $\|z_1-\gamma_1v_1\|_2 \triangleq \gamma_2 \leq \gamma_1 \gamma \leq \gamma^2$ for
 some $v_1$ in $K$. Repeating this process, we have
\begin{equation}\nonumber
 z=\sum_{j\geq 0} \gamma_j v_j
\end{equation}
where $\gamma_0=1$, $\gamma_j \leq \gamma^j$ and $v_j \in K$.

Thus for any $z \in \mathbf{R}^n$, we have
 $z=\|z\|_2\sum_{j\geq 0} \gamma_j v_j$.

For any index set $T$ with $|T| \leq \rho m$,

\begin{eqnarray*}
\sum \limits_{i \in T} |(Az)_i|^p &=& \|z\|_2^p \sum \limits_{i \in T} |(\sum \limits_{j \geq 0} \gamma_j Av_j)_i|^p \\
& \leq & \|z\|_2^p \sum \limits_{i \in T} \sum \limits_{j \geq 0} \gamma^{jp} |(Av_j)_i|^p \\
&= & \|z\|_2^p  \sum \limits_{j \geq 0} \gamma^{jp} \sum \limits_{i \in T} |(Av_j)_i|^p \\
& \leq & S\|z\|_2^p \frac{1-2\delta}{2(1-\gamma^p)}
\end{eqnarray*}

\begin{eqnarray*}
\sum \limits_{i} |(Az)_i|^p &=& \|z\|_2^p \sum \limits_{i } |(\sum \limits_{j \geq 0} \gamma_{j} Av_j)_i|^p \\
& \geq & \|z\|_2^p \sum \limits_{i}(|(Av_0)_i|^p- \sum \limits_{j \geq 1} \gamma_{j}^p |(Av_j)_i|^p) \\
& \geq & \|z\|_2^p (\sum \limits_{i} |(Av_0)_i|^p- \sum \limits_{j \geq 1} \gamma^{jp} \sum \limits_{i}|(Av_j)_i|^p) \\
& \geq & \|z\|_2^p ((1-\epsilon)S- \sum \limits_{j \geq 1} \gamma^{jp} (1+\epsilon)S) \\
& \geq &  S\|z\|_2^p \frac{1-2\gamma^p-\epsilon}{1-\gamma^p}
\end{eqnarray*}

Thus  $\sum \limits_ {i \in T^c} |(Az)_i|^p - \sum \limits_{i \in T}
|(Az)_i|^p \geq S \|z\|_2^p
\frac{2\delta-2\gamma^p-\epsilon}{1-\gamma^p}$.
For a given $\delta$, we can pick $\gamma$ and $\epsilon$ small
enough such that $\sum \limits_ {i \in T^c} |(Az)_i|^p - \sum
\limits_{i \in T} |(Az)_i|^p \geq \delta S \|z\|_2^p$.
\end{proof}

We can now establish one main result regarding the threshold of
successful recovery with $l_p$-minimization.
\begin{theorem}
For any $0<p \leq 1$, given any $\rho < \rho^*(p)$, there exist
constants $c_4$, $c_5 >0$ such that when $m \geq c_4 n$ and $n$ is
large enough, with probability $1-e^{-c_5n}$, an $m \times n$ matrix
$A$ with i.i.d. $N(0,1)$ entries has the following property: for
every $f \in \mathbf{R}^n$ and every error $e$ with its support $T$
satisfying $|T| \leq \rho m$, $f$ is the unique solution to the
$l_p$-minimization problem (\ref{eqn:lp}).
\end{theorem}

\begin{proof}
Lemma \ref{lemma:slp} indicates that $\sum _{i \in T^c} |(Az)_i|^p -
\sum _{i \in T} |(Az)_i|^p \geq \delta S \|z\|_2^p>0$ for every
non-zero $z$, then from Theorem \ref{thm:slp}, $f$ is the unique
solution to the $l_p$-minimization problem (\ref{eqn:lp}).
\end{proof}

We remark here that $\rho^*$ is a sharp bound for successful
recovery. For any $\rho>\rho^*$, from Lemma \ref{lemma:srho}, with
overwhelming probability the sum of the largest $\lceil \rho m
\rceil$ terms of $|(Az)_i|^p$'s is more than the half of the total
sum $S_1$, then Theorem \ref{thm:slp} indicates that the
$l_p$-recovery fails in this case. In fact, for any vector $f' \neq
f$, let $z=f'-f$, and let $T$ be the support of the largest $\lceil
\rho m \rceil$ terms of $|(Az)_i|^p$'s. If the error vector $e$
agrees with $|(Az)_i|^p$ on the support $T$ and is zero elsewhere,
then with large probability $\|e-Az\|_p^p$ is no greater than that
of $\|e\|_p^p$, which implies that $l_p$-minimization cannot
correctly return $f$. Proposition \ref{prop:rho} thus implies that
the threshold strictly decreases as $p$ increases. The performance
of $l_{p_1}$-minimization is better than $l_{p_2}$-minimization for
$p_1 <p_2 \leq 1$ in the sense that the sparsity requirement for the
arbitrary error vector is less strict for smaller $p$.
\section{Recovery From Error Vector With Fixed Support and Signs}\label{sec:wbd}
In Section \ref{sec:sbd}, for some $\rho>0$, we call
$l_p$-minimization successful if and only if it can recover $f$ from
any error $e$ whose support size is at most $\rho m$. Here we only
require $l_p$-minimization to recover $f$ from errors with fixed but
unknown support and signs. We will provide a sharp threshold
$\rho_w^*$ of the proportion of errors below which
$l_p$-minimization is successful.

Once the support and the signs of an error vector is fixed, the
condition of successful recovery with $l_1$-minimization from any
such error vector is the same, however, the condition of successful
recovery with $l_p$-minimization from different error vectors
differs even the support and the signs of the error is fixed. Here
we consider the worst case scenario in the sense that the recovery
with $l_p$-minimization is defined to be ``successful''
if $f$ can be recovered from any such error $e$. We characterize 
this case in Theorem \ref{thm:wb}.
Note that if there is further constraint on $e$, then the condition
of successful recovery with $l_p$-minimization may be different from
the one stated in Theorem \ref{thm:wb}.

\begin{theorem}\label{thm:wb}
Given any $p \in (0,1)$, for every $f \in \mathbf{R}^n$ and every
error $e$ with fixed support $T$ and fixed sign for each entry $e_i,
i \in T$, if $f$ is always the unique solution to $l_p$-minimization
problem (\ref{eqn:lp}), then
\begin{equation}\nonumber
\sum \limits_{i \in T^-} |(Az)_i|^p \leq \sum \limits_{i \in T^c}
|(Az)_i|^p
\end{equation}
for all $z \in \mathbf{R}^n$ where $T^-=\{i \in T: (Az)_ie_i <0\}$.

Conversely, $f$ is always the unique solution to $l_p$-minimization
problem (\ref{eqn:lp}) provided that
\begin{equation}\nonumber
\sum \limits_{i \in T^-} |(Az)_i|^p < \sum \limits_{i \in T^c}
|(Az)_i|^p
\end{equation}
for all non-zero $z \in \mathbf{R}^n$.
\end{theorem}

\begin{proof}
First part. Suppose there exists $z$ such that $\sum _{i \in T^-}
|(Az)_i|^p > \sum _{i \in T^c} |(Az)_i|^p$, let $\delta=\sum _{i \in
T^-} |(Az)_i|^p - \sum _{i \in T^c} |(Az)_i|^p>0$.

Let $e_i =0$ for every $i$ in $T^c$, let $e_i=-(Az)_i$ for every $i$
in $T^-$. For every $i$ in $T^+:= T-T^-$, let $e_i$ satisfy
$(Az)_ie_i \geq 0$. As $p \in (0,1)$, we can pick $e_i$ ($i \in
T^+$) with $|e_i|$ large enough such that $\sum_{i \in T^+}
|e_i+(Az)_i|^p - \sum_{i \in T^+} |e_i|^p< \frac{\delta}{2}$. Then
\begin{eqnarray*}
\|e+Az\|_p^p &=& \sum \limits_{i\in T^-} 0+ \sum \limits_{i\in T^+} |e_i+(Az)_i|^p+\sum \limits_{i\in T^c} |(Az)_i|^p\\
&<& \sum
\limits_{i \in T^+} |e_i|^p + \frac{\delta}{2}+\sum \limits_{i\in T^c} |(Az)_i|^p\\
&=& \sum
\limits_{i \in T^+} |e_i|^p + \frac{\delta}{2}+\sum \limits_{i \in T^-} |(Az)_i|^p-\delta\\
&=&  \|e\|_p^p - \frac{\delta}{2}.
\end{eqnarray*}
Thus $\|y-A(f-z)\|_p^p=\|e+Az\|_p^p < \|e\|_p^p=\|y-Af\|_p^p$, $f$
is not a solution to (\ref{eqn:lp}), which is a contradiction.

Second part. For any $e$ on support $T$ with fixed signs and for any
$f$, let $y=Af+e$. For any $x \neq f$, let $z=f-x$, and so
\begin{eqnarray*}
&&\|y-Ax\|_p^p =\|(y-Af)+Az\|_p^p\\
&= &\sum\limits_{i \in T^+} |e_i+(Az)_i|^p +\sum\limits_{i \in T^-} |e_i+(Az)_i|^p+ \sum\limits_{i \in T^c} |(Az)_i|^p\\
& \geq & \sum\limits_{i \in T^+} |e_i|^p +\sum\limits_{i \in T^-} (|e_i|^p-|(Az)_i|^p)+\sum \limits_{i \in T^c} |(Az)_i|^p\\
& > & \|e\|_p^p.
\end{eqnarray*}

The first inequality holds as for each $i$ in $T^+$, $(Az)_i$ has
the same sign as that of $e_i$ if not zero; and for $p \in (0,1)$,
$|e_i+(Az)_i|^p \geq |e_i|^p-|(Az)_i|^p$ holds. The second
inequality comes from the assumption that $\sum \limits_{i \in T^-}
|(Az)_i|^p < \sum \limits_{i \in T^c} |(Az)_i|^p$. Thus
$\|y-Ax\|_p^p$>$\|y-Af\|_p^p$ for all $x \neq f$.
\end{proof}

\begin{lemma}\label{lemma:rhow}
Let $X_1$, $X_2$,...,$X_m$ be i.i.d. $N(0,1)$ random variables and
$T$ be a set of indices with size $|T|=\rho m$ for some $\rho>0$.
Let $e \in \mathbf{R}^m$ be any vector on support $T$ with fixed
signs for each entry. If $\rho< \rho^*_w =\frac{2}{3}$, for every $
\epsilon >0$, when $m$ is large enough,  with probability
$1-e^{-c_6m}$ for some constant $c_6>0$, the following two
properties hold:
\begin{itemize}
\item $\frac{1}{2}\rho m (\mu-\epsilon) < \sum_{i \in T: X_ie_i <0} |X_i|^p < \frac{1}{2}\rho m (\mu+\epsilon)$ 
\item $ (1-\rho) m (\mu-\epsilon) < \sum_{i \in T^c} |X_i|^p < (1-\rho) m
(\mu+\epsilon)$.
\end{itemize}
where $\mu=E[|X|^p]$, $X \sim N(0,1)$.
\end{lemma}

\begin{proof}
Define a random variable $s_i$ for each $i$ in $T$ that is equal to
1 if $X_ie_i <0$ and equal to 0 otherwise. Then $\sum_{i \in T:
X_ie_i <0} |X_i|^p=\sum_{i \in T} |X_i|^ps_i$. $
E[|X_i|^ps_i]=\frac{1}{2}\mu$ for every $i$ in $T$ as $X_i \sim
N(0,1)$. From Chernoff bound, for any $\epsilon >0$, there exist
$d_1>0$ and $d_2>0$ such that
\begin{itemize}
\item[] $Pr[\sum_{i \in T} |X_i|^ps_i \leq \frac{1}{2}\rho m (\mu-\epsilon)] \leq e^{-d_1 m}$, 

\item[] $Pr[\sum_{i \in T} |X_i|^ps_i \geq \frac{1}{2}\rho m (\mu+\epsilon)] \leq e^{-d_2 m}.$ 
\end{itemize}
Again from Chernoff bound, there exist some constants $d_3>0$,
$d_4>0$ such that
\begin{itemize}
\item[]
$Pr[\sum_{i \in T^c} |X_i|^p \leq (1-\rho) m (\mu-\epsilon)] \leq e^{-d_3 m},$ 
\item[]
$Pr[\sum_{i \in T^c} |X_i|^p \geq (1-\rho) m (\mu+\epsilon)] \leq e^{-d_4 m}.$ 
\end{itemize}
By union bound, 
there exists some constant $c_6>0$ such that the two properties
stated in the lemma hold with probability at least $1-e^{-c_6 m}$.

\end{proof}

Lemma \ref{lemma:rhow} implies that $\sum_{i \in T: X_ie_i <0}
|X_i|^p< \sum_{i \in T^c} |X_i|^p$ holds with large probability when
$|T|=\rho m < \frac{2}{3}m$. Applying the similar net argument in
Section \ref{sec:sbd}, we can extend the result to every vector $Az$
where matrix $A^{m \times n}$ has i.i.d. Gaussian entries and $z$ is
any vector in $\mathbf{R}^n$. Then we can establish the main result
regarding the threshold of successful recovery with
$l_p$-minimization from errors with fixed support and signs.
\begin{theorem}
For any $p \in (0,1)$, given any $\rho < \frac{2}{3}$, there exist
constants $c_7$, $c_8 >0$ such that when $m \geq c_7 n$ and $n$ is
large enough, with probability $1-e^{-c_8n}$, an $m \times n$ matrix
$A$ with i.i.d. $N(0,1)$ entries has the following property: for
every $f \in \mathbf{R}^n$ and every error $e$ with fixed support
$T$ satisfying $|T| \leq \rho m$ and fixed signs on $T$, $f$ is the
unique solution to the $l_p$-minimization problem (\ref{eqn:lp}).
\end{theorem}

\begin{proof}
From lemma \ref{lemma:rhow}, applying similar arguments in the proof
of lemma \ref{lemma:slp}, we get that 
when $m \geq c_7 n$ and $n$ is large enough, with probability
$1-e^{-c_8n}$ for some $c_8>0$,
\begin{itemize}
\item $\frac{1}{2}\rho m (\mu-\epsilon)<\sum_{i \in T: (Av)_ie_i
<0} |(Av)_i|^p < \frac{1}{2}\rho m (\mu+\epsilon)$
\item $(1-\rho) m (\mu-\epsilon) < \sum_{i \in T^c} |(Av)_i|^p
< (1-\rho) m (\mu+\epsilon)$
\end{itemize}
hold for all the vectors $v$ in a $\gamma$-net $K$ at the same time.
Moreover, for any $z \in \mathbf{R}^n$, we have
 $z=\|z\|_2\sum_{j\geq 0} \gamma_j v_j$,
where $\gamma_0=1$, $v_j \in K$ for all $j$ and  $\gamma_j \leq
\gamma^j$.

Let $T^-=\{i \in T: (Az)_ie_i <0\}$. For any $i$ in $T^-$,
\begin{eqnarray*}
|(Az)_i|^p&=&\|z\|_2^p \big|(\sum_{j \geq 0} \gamma_j Av_j)_i\big|^p \\
&\leq&  \|z\|_2^p \big|(\sum_{j: (Av_j)_ie_i < 0} \gamma_j Av_j)_i\big|^p \\
&\leq& \|z\|_2^p \sum_{j: (Av_j)_ie_i < 0} \gamma^{jp} |(Av_j)_i|^p
\end{eqnarray*}
where the first inequality holds as $(Az)_ie_i <0$. Then
\begin{eqnarray*}
\sum \limits_{i \in T^-} |(Az)_i|^p&\leq &\|z\|_2^p \sum \limits_{i \in T^-}  \sum \limits_{j: (Av_j)_ie_i < 0} \gamma^{jp} |(Av_j)_i|^p \\
& \leq & \|z\|_2^p \sum \limits_{i \in T}  \sum \limits_{j: (Av_j)_ie_i < 0} \gamma^{jp} |(Av_j)_i|^p \\
& = & \|z\|_2^p  \sum \limits_{j\geq 0}  \gamma^{jp} \sum \limits_{i \in T: (Av_j)_ie_i < 0}  |(Av_j)_i|^p \\
&<& \|z\|_2^p \frac{1}{2(1-\gamma^p)}\rho m(\mu+\epsilon)
\end{eqnarray*}

\begin{eqnarray*}
&&\sum \limits_{i \in T^c} |(Az)_i|^p= \|z\|_2^p \sum \limits_{i \in T^c} |(\sum \limits_{j \geq 0} \gamma_{j} Av_j)_i|^p \\
& \geq & \|z\|_2^p \big(\sum \limits_{i\in T^c} |(Av_0)_i|^p- \sum \limits_{j \geq 1} \gamma^{jp} \sum \limits_{i\in T^c}|(Av_j)_i|^p\big) \\
& > & \|z\|_2^p \big((1-\rho)m(\mu-\epsilon)- \sum \limits_{j \geq 1} \gamma^{jp} (1-\rho)m(\mu+\epsilon)\big) \\
& \geq & \|z\|_2^p (1-\rho)m \frac{\mu-2\mu\gamma^p-\epsilon}{1-\gamma^p}
\end{eqnarray*}

Thus 
$\sum_{i \in T^c} |(Az)_i|^p- \sum _{i \in T^-} |(Az)_i|^p >
\|z\|_2^p
\frac{m\mu}{1-\gamma^p}\big(1-\frac{3}{2}\rho-2\gamma^p(1-\rho)-\frac{\epsilon}{\mu}(1-\frac{\rho}{2})\big)$. %
For any $\rho< \frac{2}{3}$, we can pick $\gamma$ and $\epsilon$
small enough such that the righthand side 
is positive. The result follows by applying Theorem \ref{thm:wb}.

\end{proof}
We remark here that $\rho^*_w$ is a sharp bound for successful
recovery in this setup. For any $\rho>\rho^*_w$, from Lemma
\ref{lemma:rhow}, with overwhelming probability that $\sum_{i \in T:
X_ie_i <0} |X_i|^p> \sum_{i \in T^c} |X_i|^p$, then Theorem
\ref{thm:wb} indicates that the $l_p$-recovery fails for some error vector
$e$ in this case.

Surprisingly, the successful recovery threshold $\rho^*$ when fixing
the support and the signs of an error vector is $\frac{2}{3}$ for
all $p$ in $(0,1)$ and is strictly less than the threshold for
$p=1$, which is 1 (\cite{Donoho06}). Thus in this case,
$l_1$-minimization has better recovery performance than that of
$l_p$-minimization ($p<1$) in terms of the sparsity requirement for
the error vector. The result seems counterintuitive, however, it
largely depends on the definition of successful recovery in terms of
worse case performance. The condition of successful recovery via
$l_1$-minimization from any error vector on the fixed support with
fixed signs is the same, while the condition of $l_p$-minimization
from different error vectors differs.

%

\vspace{0.09in}\noindent {\bf Acknowledgments:} The authors thank
anonymous reviewers for helpful comments. The research is supported
by NSF under CCF-0835706.

\bibliographystyle{IEEEtran}

\begin{thebibliography}{10}
\providecommand{\url}[1]{#1}
\csname url@samestyle\endcsname
\providecommand{\newblock}{\relax}
\providecommand{\bibinfo}[2]{#2}
\providecommand{\BIBentrySTDinterwordspacing}{\spaceskip=0pt\relax}
\providecommand{\BIBentryALTinterwordstretchfactor}{4}
\providecommand{\BIBentryALTinterwordspacing}{\spaceskip=\fontdimen2\font plus
\BIBentryALTinterwordstretchfactor\fontdimen3\font minus
  \fontdimen4\font\relax}
\providecommand{\BIBforeignlanguage}[2]{{%
\expandafter\ifx\csname l@#1\endcsname\relax
\typeout{** WARNING: IEEEtran.bst: No hyphenation pattern has been}%
\typeout{** loaded for the language `#1'. Using the pattern for}%
\typeout{** the default language instead.}%
\else
\language=\csname l@#1\endcsname
\fi
#2}}
\providecommand{\BIBdecl}{\relax}
\BIBdecl

\bibitem{DoT05}
D.~L. Donoho and J.~Tanner, ``Sparse nonnegative solution of underdetermined
  linear equations by linear programming,'' in \emph{Proc. Natl. Acad. Sci.
  U.S.A.}, vol. 102, no.~27, 2005, pp. 9446--9451.

\bibitem{Don06}
D.~Donoho, ``Compressed sensing,'' \emph{IEEE Trans. Inf. Theory}, vol.~52,
  no.~4, pp. 1289--1306, April 2006.

\bibitem{CaT05}
E.~Cand\`{e}s and T.~Tao, ``Decoding by linear programming,'' \emph{IEEE Trans.
  Inf. Theory}, vol.~51, no.~12, pp. 4203--4215, Dec. 2005.

\bibitem{CaT06}
------, ``Near-optimal signal recovery from random projections: Universal
  encoding strategies?'' \emph{IEEE Trans. Inf. Theory}, vol.~52, no.~12, pp.
  5406--5425, Dec. 2006.

\bibitem{SXH08}
M.~Stojnic, W.~Xu, and B.~Hassibi, ``Compressed sensing - probabilistic
  analysis of a null-space characterization,'' in \emph{Proc. ICASSP}, 2008,
  pp. 3377--3380.

\bibitem{WM08}
J.~Wright and Y.~Ma, ``Dense error correction via $l^1$ minimization,''
  \emph{Preprint}, 2008.

\bibitem{Chartrand07}
R.~Chartrand, ``Exact reconstruction of sparse signals via nonconvex
  minimization,'' \emph{Signal Process.Lett.}, vol.~14, no.~10, pp. 707--710,
  2007.

\bibitem{Chartrand072}
------, ``Nonconvex compressed sensing and error correction,'' in \emph{Proc.
  ICASSP}, 2007.

\bibitem{SCY08}
R.~Saab, R.~Chartrand, and O.~Yilmaz, ``Stable sparse approximations via
  nonconvex optimization,'' in \emph{Proc. ICASSP}, 2008.

\bibitem{DG09}
M.~E. Davies and R.~Gribonval, ``Restricted isometry constants where
  \mbox{$l_p$} sparse recovery can fail for $0 < p \leq 1$,'' \emph{IEEE Trans.
  Inf. Theory}, vol.~55, no.~5, pp. 2203--2214, 2009.

\bibitem{FL09}
S.~Foucart and M.-J. Lai, ``Sparsest solutions of underdetermined linear
  systems via $l_q$-minimization for $0<q \leq 1$,'' \emph{Applied and
  Computational Harmonic Analysis}, vol.~26, no.~3, pp. 395 -- 407, 2009.

\bibitem{DMT07}
C.~Dwork, F.~McSherry, and K.~Talwar, ``The price of privacy and the limits of
  lp decoding,'' in \emph{Proc. STOC}, 2007, pp. 85--94.

\bibitem{Ledoux01}
M.~Ledoux, Ed., \emph{The Concentration of Measure Phenomenon}.\hskip 1em plus
  0.5em minus 0.4em\relax American Mathematical Society.

\bibitem{Donoho06}
D.~Donoho, ``High-dimensional centrally symmetric polytopes with neighborliness
  proportional to dimension,'' \emph{Discrete Comput. Geom.}, 2006.

\end{thebibliography}

\end{document}